\documentclass[10pt,twocolumn,conference,final]{IEEEtran}
\usepackage{epsf,psfrag,amssymb,amsfonts,amsmath,graphicx,times,color,subfigure,stfloats}
\usepackage[square, comma, sort&compress, numbers]{natbib}
\usepackage[mathscr]{eucal}

\def\boxit#1{\vbox{\hrule\hbox{\vrule\kern3pt
        \vbox{\kern3pt#1\kern3pt}\kern3pt\vrule}\hrule}}
\def\dag{{\cal y}}
\def\reals{ { {\rm  I \kern-0.15em R }  } }
\def\complex{ {\,{{\rm C} \kern-0.50em \raise0.20ex {  |}}\, }}

\def\Rbf{{\bf R}}

\def\be{\begin{equation}}
\def\ee{\end{equation}}

\def\scalefig#1{\epsfxsize #1\textwidth}
\def\Rxx{\Rbf_{\ssstyle X\kern-.1em X}}

\let\ssstyle=\scriptscriptstyle

\def\etal{{\it et al. \/}}

\def\ie{{\it i.e.,\ \/}}
\def\Kout{\setbox1=\hbox{\Huge\bf K}\hbox to
1.05\wd1{\hspace{.05\wd1}
}}
\def\Sout{\setbox1=\hbox{\Huge\bf S}\hbox to 1.05\wd1{\hspace{.05\wd1}
}}

\makeatletter
\def\ps@headings{%
\def\@oddhead{\mbox{}\scriptsize\rightmark \hfil \thepage}%
\def\@evenhead{\scriptsize\thepage \hfil \leftmark\mbox{}}%
\def\@oddfoot{}%
\def\@evenfoot{}}
\makeatother 

\def\ie{{\it i.e.,\ \/}}

\newtheorem{property}{Property}
\newtheorem{theorem}{Theorem}
\newtheorem{lemma}{Lemma}
\def\scalefig#1{\epsfxsize #1\textwidth}

\newtheorem{definition}{Definition}
\def\te{\tilde{e}}
\def\tE{\tilde{E}}

\newcounter{algleo}
\newlength{\lefttab}
\newenvironment{algleo}%
  {\trivlist
   \topsep=0pt\itemsep=0pt
   \def\li{\item\refstepcounter{algleo}\makebox[0pt][r]{\thealgleo\hspace{\lefttab}\hspace{-0.35cm}}
   \hangafter1\hangindent1em\noindent}%
   \def\linonumber{\item\makebox[0pt][r]{\hspace{\lefttab}}
   \hangafter1\hangindent1em\noindent}%
   \addtolength{\lefttab}{1.25em}
   \leftskip=\lefttab}%
  {\endtrivlist}
\def\If{{\bf if }}
\def\Then{{\bf then }}
\def\Else{{\bf else }}
\def\While{{\bf while }}

\def\Foreach{{\bf for each }}

\def\Do{{\bf do }}
\def\End{{\bf end}}

\begin{document}

\title{Dynamic Shortest Path Algorithms for Hypergraphs}

\author{\IEEEauthorblockN{J. Gao$^\dag$, Q. Zhao$^\dag$, W. Ren$^\ddag$, A. Swami$^\S$, R.Ramanathan$^\P$, A. Bar-Noy$^\sharp$ }
\IEEEauthorblockA{$^\dag$UC Davis, $^\ddag$Microsoft, $^\S$Army Research Lab, $^\P$Raytheon BBN Technologies,
$^\sharp$City University of New York}}


\maketitle

\begin{abstract}
A hypergraph is a set $V$ of vertices and a set of non-empty subsets of $V$, called hyperedges.
Unlike graphs, hypergraphs can capture higher-order interactions in social and communication networks that go beyond
a simple union of pairwise relationships. In this paper, we consider the shortest path problem in hypergraphs.
We develop two algorithms for finding and maintaining the shortest hyperpaths in a dynamic network with both weight and
topological changes. These two algorithms are the first addressing the fully dynamic shortest path problem in a general
hypergraph. They complement
each other by partitioning the application space based on the nature of the change dynamics and the type of the hypergraph.
We analyze the time complexity of the proposed algorithms and perform simulation experiments for both random geometric hypergraphs and the
Enron email data set. The latter illustrates the application of the proposed algorithms in social networks for identifying
the most important actor based on the closeness centrality metric.
\end{abstract}

\section{Introduction}
A\footnotetext{This work was supported by the Army Research Laboratory NS-CTA under Grant W911NF-09-2-0053.}
graph is a basic mathematical abstraction for modeling networks,
in which nodes are represented by vertices and pairwise relationships
are represented by edges between vertices. A graph is thus given by
a vertex set $V$ and an edge set $E$ consisting of cardinality-2
subsets of $V$. A hypergraph is a natural extension of a graph
obtained by removing the constraint on the cardinality of an edge:
any non-empty subset of $V$ can be an element (a hyperedge) of the edge set
$E$ (see Fig~\ref{fig:HG_SC_example}). It thus captures group behaviors and higher-dimensional
relationships in complex networks that are more than a simple union
of pairwise relationships. Examples include communities and
collaboration teams in social networks, document clusters in
information networks, and cliques, neighborhoods, and multicast groups
in communication networks.

\begin{figure}[htbp]
\centering
\begin{psfrags}
\psfrag{1}[c]{$v_1$}
\psfrag{2}[c]{$v_2$}
\psfrag{3}[c]{$v_3$}
\psfrag{4}[c]{$v_4$}
\psfrag{5}[c]{$v_5$}
\psfrag{6}[c]{$v_6$}
\psfrag{7}[c]{$v_7$}
\psfrag{8}[c]{$v_8$}
\psfrag{9}[c]{$v_9$}
\scalefig{0.3}\epsfbox{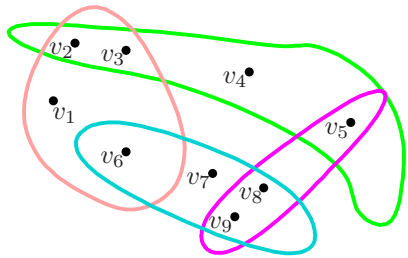}
\end{psfrags}
\caption{An example hypergraph with $4$ hyperedges:
$(v_1,v_2, v_3, v_6)$, $(v_2,v_3,v_4,v_5)$, $(v_6,v_7,v_8,v_9)$, and $(v_5,v_8,v_9)$. }
\label{fig:HG_SC_example}
\end{figure}

While the concept of hypergraph has been around
since 1920's (see, for example, \cite{berge1976graphs}), many well-solved algorithmic problems in graph theory remain
largely open under this more general model. In this paper, we
address the shortest path problem in hypergraphs.

\subsection{Shortest Path Problem in Graphs}
The shortest path problem is perhaps one of the most basic problems in
graph theory. It asks for the shortest path between two vertices or from a
source vertex to all the other vertices (\ie the single-source
version or the shortest path tree). Depending on whether the edge weights can be negative,
the problem can be solved via Dijkstra's
algorithm or Bellman-Ford algorithm~\cite{Kleinberg&Tardos:Algo_Design}.
This basic problem finds diverse applications in communication networks,
operational research, plant and facility layout, and VLSI
design~\cite{Chen:96CSUR}.

The dynamic version of the shortest path problem is to maintain the shortest path tree
without recomputing from scratch during a sequence of changes to the graph.
A typical change to a graph includes weight increase, weight decrease, edge insertion, and edge deletion.
The last two types of changes model network topological changes, but they can be
conceptually considered as special cases of weight changes by allowing weight to be infinity. Thus,
if the sequence of changes contains only weight increase and edge deletion, we call it a decremental problem;
if it contains only weight decrease and edge insertion, we call it an incremental problem.
Otherwise, we have a fully dynamic problem. If multiple edges change simultaneously, then it is called a batch problem.

There have been a number of studies of the dynamic shortest path problem in graphs. Ramalingam and Reps~\cite{ramalingam1996computational},
Frigioni \etal~\cite{Frigioni1998,frigioni2000fully}, and Narvaez \etal~\cite{narvaez2000new} proposed several
algorithms for the single-change problem. The batch problem was considered
in~\cite{ramalingam1996incremental,narvaez2000new,bauer2009batch}. Comprehensive experiments on the comparison of
different batch algorithms can be found in~\cite{bauer2009batch,taoka2007performance}.

\subsection{Shortest Path Problem in Hypergraphs}

Both the static and dynamic shortest path problems have a corresponding version
in hypergraphs. The static shortest hyperpath problem was considered by
Knuth~\cite{Knuth:77IPL} and Gallo \etal~\cite{Gallo&Etal:93DAM}, in which Dijkstra's algorithm was
extended to obtain the shortest hyperpaths. Knuth's algorithm is for a special class of hypergraphs
while Gallo's algorithm is for a general hypergraph.
Ausiello \etal proposed a dynamic shortest hyperpath algorithm for directed hypergraphs,
considering only the incremental problem with the weights of all hyperedges limited to a finite
set of numbers~\cite{ausiello1990dynamic,ausiello1992optimal}.
A dynamic algorithm for the batch problem in a special class of hypergraphs was developed in \cite{ramalingam1996incremental}.

With the exception of the above few studies, the shortest hyperpath
problem remains largely unexplored. To the best of our knowledge, no algorithms exist
for the fully dynamic problem in a general hypergraph.

In this paper, we develop two fully dynamic shortest path algorithms for general hypergraphs.
These two algorithms complement each other, with each preferred
in different types of hypergraphs and dynamics.

Referred to as the HyperEdge based Dynamic Shortest Path algorithm (HE-DSP), the first algorithm
is an extension of the dynamic Dijkstra's algorithm for graphs to hypergraphs (parallel to Gallo's extension
of the static Dijkstra's algorithm to hypergraphs in~\cite{Gallo&Etal:93DAM}). The extension of the dynamic
Dijkstra's algorithm to hypergraphs is more involved than that of the static Dijkstra's algorithm. This is due to
the loss of the tree structure (in the original graph sense) in the collection of the shortest hyperpaths from a source to all other vertices.
Since the dynamic Dijkstra's algorithm relies on the tree structure to update the shortest paths after an incremental
change (weight increase or edge deletion), special care needs to be given when extending it to hypergraphs.

The second algorithm is rooted in the idea of Dimension Reduction and is referred to as DR-DSP.
The basic idea is to reduce the problem to
finding the shortest path in the underlying graph of the hypergraph. The underlying graph of a hypergraph
has the same vertex set and has an edge between two vertices if and only if
there is at least one hyperedge containing these two vertices in the original hypergraph. The weight of an edge
in the underlying graph is defined as the minimum weight among all hyperedges containing the two vertices of
this edge. The shortest hyperpath in the hypergraph can thus be obtained from the shortest path in the underlying
graph by substituting each edge along the shortest path with the hyperedge that lent its weight to this edge.
The correctness and advantage of this algorithm are readily seen: the definition of weight in the
underlying graph captures the minimum cost offered by all hyperedges in choosing a path between two vertices, thus
ensuring the correctness of the algorithm; the reduction of a hypergraph to its underlying graph removes many
hyperedges from consideration when finding the shortest path, leading to efficiency and agility to dynamic changes.

HE-DSP is more efficient in hypergraphs that are densely connected through high-dimensional hyperedges
and for network dynamics where changes often occur to hyperedges that are not on the current
shortest hyperpaths. DR-DSP has lower complexity when hyperedge changes often lead to changes in
the shortest hyperpaths. This is usually the case in networks where hyperedges in the shortest hyperpaths
are more prone to changes due to attacks, frequent use, or higher priority in maintenance and upgrade. Furthermore, DR-DSP leads to an alternative algorithm for solving the static shortest hyperpath problem when the dynamic problem degenerates to the static problem. It
has the same complexity as Gallo's algorithm for a general hypergraph and lower complexity
for simplicial complexes (a special class of hypergraphs whose hyperedge set is closed under the subset operation).
We also point out that both proposed algorithms apply to directed hypergraphs with minor modifications in their implementation details.

A detailed time complexity analysis of these two algorithms is provided to demonstrate their performance
in the worst-case change scenario. Using a random geometric
hypergraph model and a real data set of a social network (Enron email data set), we study the average performance of these two algorithms
in different scenarios and demonstrate the partition of the application space between these two algorithms.
In the experiment with Enron email data set, the proposed algorithms successfully identified the most important actor
in this social network using the closeness centrality metric.

\subsection{Applications}

Shortest path computations on hypergraphs can be applied to
communication as well as social networks. An example application in
wireless communications, in particular, for multihop wireless networks, is
in {\em opportunistic routing} schemes such as ExOR~\cite{biswas2005exor},
GeRaF~\cite{zorzi2003geographic}, and MORE\cite{Chachulski&etal:07}.
In such schemes, any receiver of a packet is eligible to forward
the packet. Receivers typically execute a protocol amongst themselves
to decide who should forward it. This naturally leads to a hypergraph model
where a node and its neighbors form a hyperedge. The cost of each hyperedge can be defined
based on
the cardinality of the hyperedge to capture the success rate of forwarding (lower the cardinality, lesser the chance that
at least one of the nodes successfully receives the packet) and the associated overhead (higher the cardinality, higher
the energy consumption and the overhead in choosing the forwarding node). A shortest hyperpath from the
source to the destination is thus a better route than merely the traditional
shortest path. And as the network topology changes, a dynamic algorithm is required to
maintain the shortest hyperpath.

In social networks, information (results, event reports,
opinions, rumors, \emph{etc.}) propagates through diverse communication means including direct links (e.g.,
gestures, optical, satcom, regular phone call), social media (e.g.,
Facebook, Twitter, blogs), mailing lists, and newsgroups. Such a
network may be modeled as a hypergraph with the weight of
a hyperedge reflecting the cost, credibility, and/or delay for disseminating
information among all vertice of this hyperedge. In particular, the weight of a hyperedge
can capture the unique effect on the information after it passes through a group of people.
For instance, a result can be discussed by overlapping blog collaboration
networks as it spreads, and often the discussion yields a better
result than if it only spreads through individuals. The minimum cost information passing in social networks
can thus be modelled as a shortest hyperpath problem.

Another potential application is that of finding the most important
actor in a social network. Under a graph model of social networks,
the relative importance of a vertex can be measured by its
betweenness and closeness centrality indices. The former is defined
based on the number of shortest paths that pass through this vertex,
and the latter, the total weight of the shortest paths from this
vertex to all the other vertices~\cite{Wasserman&Faust:Social_Net_Ana}.
In a social network exhibiting hyper-relationships, betweenness and
closeness centrality, based on the shortest hyperpaths, would
be better indicators of the relative importance of each actor.
In Sec.~\ref{sec:simu}, we apply the proposed shortest hyperpath algorithms to
the Enron email data set. We propose a weight function that leads to the successful
identification of the CEO of Enron as the most important actor under the
closeness centrality metric. The distance of each person in the data set to
the CEO along the resulting shortest hyperpaths closely reflects the position
of the person within the company.

\section{Background on Dynamic Shortest Path Algorithms for Graphs}
\label{sec:BG}

In this section, we present the basic ideas of the dynamic shortest path algorithms developed for graphs
in~\cite{frigioni2000fully}.
Some basic techniques in updating and maintaining the shortest path tree will be borrowed in later sections
when we develop dynamic shortest hyperpath algorithms.

\subsection{Dynamic Shortest Path Problem}
A change $\delta$ on a graph $G=(V,E)$ 
corresponds to one edge modification.
There are four types of changes: weight increase, weight decrease, edge insertion, and edge deletion.
Weight increase and edge deletion can be similarly treated (with small differences in the required data structures
which will be omitted for simplicity), so can weight decrease and edge insertion. The dynamic algorithms are
thus presented only for weight increase and weight decrease.

Given a graph $G$, a source node $s$, and a sequence of changes $C=\{\delta_1,\delta_2,\ldots,\delta_l\}$ on $G$,
the dynamic shortest path problem is to find the shortest paths from $s$ to all nodes in each new graph after change $\delta_i$.

In the following, $D[v]$ denotes the distance of a vertex $v$ to the source $s$, $P[v]$ the parent of $v$ in the shortest path tree.
A vertex $v$ is called an affected vertex if $D[v]$ or $P[v]$ or both change in the new shortest path tree.
An edge is called an affected edge if it contains an affected vertex.

\subsection{Weight Decrease}
Consider that the weight of edge $(\check{u},\check{v})$ decreases to $w_{new}$.
Without loss of generality, assume that $D[\check{u}]\leq D[\check{v}]$.
It is not difficult to see that $\check{u}$ will not be affected by this change.
The dynamic algorithm starts with
determining whether $\check{v}$ will be affected by simply checking the inequality
\begin{equation}
D[\check{u}]+w_{new}<D[\check{v}].
\label{eq:ineq}
\end{equation}
If the inequality does not hold, then this edge with the decreased weight does not provide a shorter path for $\check{v}$;
the algorithm ends and the shortest path tree remains unchanged.
If the inequality holds, then $\check{v}$ is affected; its new shortest path from $s$ must go through edge $(\check{u},\check{v})$
and $D[\check{v}]$ reduces to $D[\check{u}]+w_{new}$. We put $\check{v}$ in a priority queue\footnote{A priority queue is an abstract data type
with the following access protocol: only the highest-priority element can be accessed. Basic operations of a priority queue
include Enqueue (add a new item to the queue), Dequeue (remove the item with the highest priority and return this item),
Update (change the priority of one item in the queue), and Peek (obtain the value of the item with the highest priority).
Standard implementations of a priority queue with different time complexities include
array, link list, Binary heap, and Fibonacci heap~\cite{Dale:06}.} $Q$, and the rest of
the procedure is similar to Dijkstra's algorithm:
dequeue the node $z$ with the minimum distance from $Q$, update the distances of its neighbors, update $Q$ by
inserting the new affected vertices
among the neighbors to Q and update the ranks of others based on the updated distances. The procedure iterates until $Q$ is empty.
A pseudo code presentation of the basic steps is given below.

\setcounter{algleo}{0}
\begin{algleo}
\linonumber {\bf Graph: Weight Decrease$(\check{u},\check{v},w_{new})$}.
\linonumber {\bf Step0 (Update the graph)}
\begin{algleo}
\li $w(\check{u},\check{v})\leftarrow w_{new}$
\end{algleo}
\linonumber {\bf Step1 (Determine the affected vertex in $(\check{u},\check{v})$)}
\begin{algleo}
\li $x\leftarrow \textrm{argmin}_{q\in\{\check{u},\check{v}\}}\{D[q]\}$; $y\leftarrow \textrm{argmax}_{q\in\{\check{u},\check{v}\}}\{D[q]\}$
\li \If $D[x]+w_{new}<D[y]$ \Do
\begin{algleo}
\li $D[y]\leftarrow D[x]+w_{new}$; $P[y]\leftarrow x$
\li Enqueue$(Q, \langle y,D[y]\rangle)$
\end{algleo}
\li \End
\end{algleo}
\linonumber {\bf Step2 (Iteratively update all affected vertices)}
\begin{algleo}
\li \While NonEmpty $(Q)$ \Do
\begin{algleo}
\li $\langle z,D[z]\rangle \leftarrow$ Dequeue$(Q)$
\li \Foreach $v\in V$ s.t. $(z,v)\in E$
\begin{algleo}
\li \If $D[v]>D[z]+w(z,v)$ \Then
\begin{algleo}
\li $D[v]\leftarrow D(z)+w(z,v)$; $P[v]\leftarrow z$
\li Enqueue or Update$(Q,\langle v,D[v]\rangle)$
\end{algleo}
\li \End; \End; \End
\end{algleo}
\end{algleo}
\end{algleo}
\end{algleo}

\subsection{Weight Increase}
Consider that the weight of edge $(\check{u},\check{v})$ increases to $w_{new}$.
Again, assume that $D[\check{u}]\leq D[\check{v}]$.
If $(\check{u},\check{v})$ is not an edge in the shortest path tree, then none of the vertices will be affected, the shortest path tree remain unchanged.
Otherwise, the descendants, and only the descendants of this edge in the shortest path tree may be affected. For these vertices, some of them
will have increased distances, some of them will go through an alternative path with the same distance (but changed parent),
while the rest will not be affected. In order to classify the vertices into these three categories,
we introduce the coloring idea in Frigioni's algorithm \cite{frigioni2000fully}:
\begin{itemize}
\item[\textbf{(1)}]  $v$ is colored {\em white} if neither $D[v]$ nor $P[v]$ needs to be changed.
\item[\textbf{(2)}]  $v$ is colored {\em pink} if $P[v]$ needs to be changed but $D[v]$ remains the same.
\item[\textbf{(3)}]  $v$ is colored {\em red} if $D[v]$ increases.
\end{itemize}

It is not difficult to see that if a vertex $v$ is white or pink, all its descendants in the shortest path tree are white;
if $v$ is red, all its descendants are either red or pink. Therefore the coloring procedure is clear:
we first determine whether $\check{v}$ is pink or red by checking whether there is an alternative shortest path with the same distance for $\check{v}$
(note that $\check{v}$ cannot be white
due to the weight change of edge $(\check{u},\check{v})$ that is on its current shortest path);
if such a path exists, then we color $\check{v}$ pink and the algorithm ends, otherwise we color it red and put all its children in a priority queue $M$.
The procedure then iterates for each vertex in $M$ according to an increasing order of the vertex distances.

After the coloring process, we only need to deal with the red vertices. For each red vertex $z$, we initialize its distance
with the distance of the shortest path through one of its non-red neighbors and put $z$ in another priority queue $Q$
(if no non-red neighbor exists, we initialize it with $\infty$). After this, the procedure is similar to Step~2 in the Graph: Weight Decrease
algorithm: at each iteration, we extract the vertex at the top of $Q$ and update its neighbors and $Q$ until $Q$ is empty.

\setcounter{algleo}{0}
\begin{algleo}
\linonumber {\bf Graph: Weight Increase}$(\check{u},\check{v},w_{new})$.
\linonumber {\bf Step0 (Update the graph)}
\begin{algleo}
\li $w(\check{u},\check{v})\leftarrow w_{new}$
\end{algleo}
\linonumber {\bf Step1 (Determine the affected vertex in $(\check{u},\check{v})$)}
\begin{algleo}
\li $x\leftarrow \textrm{argmin}_{q\in \{\check{u},\check{v}\}}\{D[q]\}$
\li $y\leftarrow \textrm{argmax}_{q\in \{\check{u},\check{v}\}}\{D[q]\}$
\li \If $P[y]=x$ \Then
\begin{algleo}
\li Enqueue$(M, \langle y,D[y]\rangle)$
\end{algleo}
\end{algleo}
\linonumber {\bf Step 2 (Coloring Process)}
\begin{algleo}
\li \While NonEmpty($M$)
\begin{algleo}
\li $\langle z, D[z]\rangle\leftarrow$ Dequeue($M$)
\li \If $\exists$ $nonred$ $q\in V$ s.t. $D[q]+w(q,z)=D[z]$
\begin{algleo}
\li \Then z is pink
\li \Else z is red; Enqueue($M$, all $z$'s children)
\end{algleo}
\li \End; \End
\end{algleo}
\end{algleo}
\linonumber {\bf Step3.a (Initialize the distance vector for red vertices)}
\begin{algleo}
\li \Foreach $red$ vertex $z$ \Do
\begin{algleo}
\li \If $z$ has no $nonred$ neighbor
\begin{algleo}
\li \Then $D[z]\leftarrow +\infty$; $P[z]\leftarrow$ Null
\li \Else
\begin{algleo}
\li let $u$ be the $best$ $nonred$ $neighbor$ of z
\li $D[z]\leftarrow D[u]+w(u,z)$; $P[z]\leftarrow u$
\li Enqueue$(Q,\langle z,D[z] \rangle)$
\end{algleo}
\li \End; \End; \End
\end{algleo}
\end{algleo}
\end{algleo}
\linonumber {\bf Step3.b: Step2 of Graph: Weight Decrease}
\end{algleo}

The worst-case time complexity for one edge change (either weight decreasing or increasing)
is $O(|\delta|\log|\delta|+\|\delta\|)$, where $|\delta|$ denotes the number of affected vertices
and $\|\delta\|$ the total number of both affected vertices and affected edges.

\section{Dynamic Shortest Hyperpath Problem}
\label{sec:PF}

We introduce some basic concepts of hypergraph~\cite{berge1976graphs} and define the static and the dynamic
shortest hyperpath problems. Some basic properties of the shortest hyperpaths are established and will be used in developing
the dynamic algorithms in subsequent sections.

\subsection{Hypergraph and Hyperpath}
Let $V$ be a finite set and $E$ a family of subsets of $V$. If for all elements $e_i\in E$, the following conditions are satisfied:
\[e_i\neq \emptyset,~~~~~~ \cup_{e_i\in E} \, e_i=V,\]
then the couple $H=(V,E)$ is called a \textit{(undirected) hypergraph}.
Each element $v\in V$ is called a \textit{vertex} and each element $e\in E$ a \textit{hyperedge}.

A \textit{weighted undirected hypergraph} is a triple $H=(V,E,w)$ with $w: E\rightarrow \{R^+\cup \{0\}\}$
being a nonnegative weight function defined for each hyperedge in $E$.

In a hypergraph, a hyperpath is defined as follows.

\begin{definition}
A {\em hyperpath} between two vertices $u$ and $v$ is a
sequence of hyperedges $\{e_0,e_1,\ldots,e_{m}\}$ such that $u \in e_0$,
$v \in e_{m}$, and $e_i \cap e_{i+1}\neq \emptyset$ for
$i=0,...,m-1$. A hyperpath is {\em simple} if non-adjacent hyperedges in the path are non-overlapping,
\ie $e_i\cap e_j=\emptyset,\forall j\neq i,i\pm1$.
\end{definition}

Let $L_e=\{e_0,\ldots,e_{m}\}$ be a hyperpath in a weighted hypergraph $H$. We define the weight of $L_e$ as:
\[
w(L_e)=\sum_{i=0}^{m} w(e_i).
\]

\subsection{Shortest Hyperpath and Relationship Tree}

Given two vertices $u$ and $v$, a natural question is to find the shortest hyperpath (in terms of the path weight)
from $u$ to $v$. Since the weight function is nonnegative, it suffices to consider only simple hyperpaths.
If the shortest hyperpath is not simple, we can always generate a simple hyperpath without increasing the
weight by deleting all the hyperedges between two overlapping non-adjacent hyperedges.

The dynamic shortest hyperpath problem can be similarly defined for a sequence $C=\{\delta_1,\delta_2,\ldots,\delta_l\}$ of
hyperedge changes. Hyperedge changes have the same four types as edge changes in a graph: weight increase,
weight decrease, hyperedge insertion, and hyperedge deletion.
Similarly, weight increase and hyperedge deletion will be treated together, so are weight decrease and hyperedge insertion.

In this paper, we consider the single-source shortest hyperpath problem: find the shortest hyperpaths from a
given source $s$ to all other vertices. The presentation of the paper focuses on undirected hypergraphs. However,
the two proposed dynamic algorithms apply to directed hypergraphs with minor modifications in their implementation details.

Below, we establish a basic property of shortest hyperpaths.

\begin{lemma}\label{lma:SP}
Let $L=\{e_1, e_2,\ldots,e_l\}$ be a shortest hyperpath from $s\in e_1$ to $z\in e_l$. Then for any vertex $v\in e_i\cap e_{i+1}$,
the hyperpath $L_v=\{e_1,e_2,\ldots,e_i\}$ is a shortest hyperpath from $s$ to $v$.
Furthermore, for any two vertices $u,v\in e_i\cap e_{i+1}$ (if there exist at least two vertices in $e_i\cap e_{i+1}$),
$D[u]=D[v]$.
\end{lemma}
\begin{proof}
We will prove by contradiction. Assume that
 $L_v=\{e_1,e_2,\ldots,e_i\}$ is not a shortest hyperpath for $v$. Then there exists a different hyperpath $L'_v=\{e'_1,e'_2,\ldots,e'_k\}$ with $w(L'_v)<w(L_v)$. Then consider the hyperpath $L'=\{e'_1,e'_2,\ldots,e'_k,e_{i+1},e_{i+2},\ldots,e_{l}\}$, we have $w(L')<w(L)$ which contradicts the fact that $L$ is a shortest hyperpath to $z$. This completes the proof for the first part of the lemma. Furthermore, for any two nodes $u,v\in e_i\cap e_{i+1}$, since $L_v$ is the shortest hyperpath for both vertices, $D[v]=w(L_v)=D[u]$.
\end{proof}

Next, we introduce the concept of relationship tree that is needed in the proposed dynamic shortest hyperpath algorithm HE-DSP.
Since two adjacent hyperedges in a hyperpath may
overlap at more than one vertex, the shortest hyperpaths from $s$ to all other vertices do not generally form a tree in the original graph sense.
For the development of the dynamic shortest hyperpath algorithms, we introduce the concept of \emph{relationship tree}
to indicate the parent-child relationship along shortest hyperpaths. The concept can be easily explained in the example given in
Fig~\ref{fig:Rtree}. Let $\{e_1,e_2\}$ be a shortest hyperpath from $s$ to $v_4$. By Lemma~\ref{lma:SP}, $\{e_1\}$ is a shortest hyperpath for both $v_1$ and $v_2$. As illustrated in Fig~\ref{fig:Rtree},
there are $4$ possible relationship trees to indicate the parent-child relationship in these shortest hyperpaths. We will show
in Sec.~\ref{sec:DEDSP} that the choice of the relationship tree does not affect the correctness or performance of the proposed
algorithm HE-DSP.

\begin{figure}[htbp]
\centering
\begin{psfrags}
\psfrag{v1}[l]{$v_1$}
\psfrag{v2}[l]{$v_2$}
\psfrag{v3}[l]{$v_3$}
\psfrag{v4}[l]{$v_4$}
\psfrag{s}[c]{$s$}
\psfrag{e1}[l]{{\small $~e_1$}}
\psfrag{e2}[l]{{\small $~e_2$}}
\psfrag{S}[c]{$s~$}
\psfrag{V1}[c]{$v_1$}
\psfrag{V2}[c]{$v_2$}
\psfrag{V3}[c]{$v_3$}
\psfrag{V4}[c]{$v_4$}
\scalefig{0.3}\epsfbox{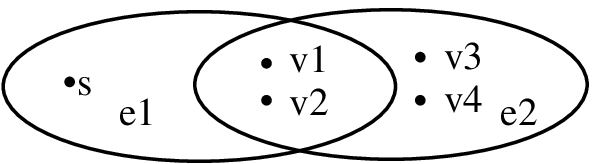}
\vspace{0.5em}
\scalefig{0.45}\epsfbox{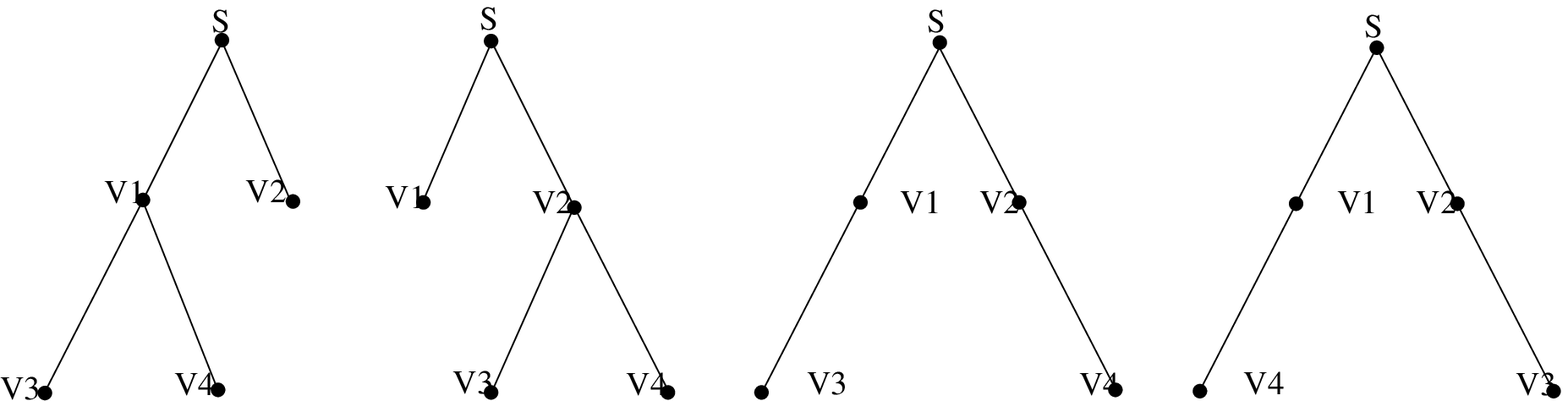}
\end{psfrags}
\caption{Hyperpaths and the associated relationship trees. }
\label{fig:Rtree}
\end{figure}

Similar notations are used for dynamic shortest hyperpath algorithms: $D[v]$ denotes the distance of a vertex $v$ to the source $s$ on
the shortest hyperpath, $P[v]$ the parent of $v$ in the chosen relationship tree associated with the shortest hyperpaths. A new notation
is $E[v]$, the hyperedge containing $v$ and $P[v]$ on the shortest hyperpath (\ie the hyperedge that leads to $v$ from $P[v]$ on the
shortest hyperpath). When it is necessary to distinguish the shortest distance before and after a weight change, $d[v]$ denotes the shortest distance before the change, $d'[v]$ the shortest distance after the change, and $D[v]$ the actual value stored in the data structure during the execution of the algorithm.

\section{Hyperedge Based Dynamic Shortest Path Algorithm}
\label{sec:DEDSP}

In this section, we propose HE-DSP. It is an extension of the dynamic Dijkstra's algorithm to hypergraphs.
The extension is more complex than Gallo's extension
of the static Dijkstra's algorithm, since the dynamic Dijkstra's algorithm relies on the tree structure of the shortest paths,
a structure no longer there for the shortest hyperpaths.

\subsection{Hyperedge Weight Decrease}
Consider that the weight of a hyperedge $\check{e}$ decreases to $w_{new}$.
Similar to the case for graphs, we know that the vertex $x\in \check{e}$ with $D[x]=\min_{v\in \check{e}}\{D[v]\}$ will not be affected.
We then check weather the other vertices in $\check{e}$ are affected by checking the inequality given in \eqref{eq:ineq},
and put all the affected vertices into a priority queue $Q$. The rest of the procedure is similar to that for graphs,
only when we update the distance of a vertex, we check all the hyperedges that contain this vertex.

\setcounter{algleo}{0}
\begin{algleo}
\linonumber {\bf HE-DSP: Weight Decrease$(\check{e},w_{new})$}.
\linonumber {\bf Step0 (Update the hypergraph)}
\begin{algleo}
\li $w(\check{e})\leftarrow w_{new}$
\end{algleo}
\linonumber {\bf Step1 (Determine the affected vertices in $e$)}
\begin{algleo}
\li $x\leftarrow \textrm{argmin}_{v\in \check{e}}\{D[v]\}$
\li \Foreach $v\in \check{e}$ such that $D[x]+w_{new}<D[v]$ \Do
\begin{algleo}
\li $D[v]\leftarrow D[x]+w_{new}$; $P[v]\leftarrow x$; $E[v]\leftarrow \check{e}$
\li Enqueue$(Q, \langle v,D[v]\rangle)$
\end{algleo}
\li \End
\end{algleo}
\linonumber {\bf Step2 (Iteratively enqueue and update affected vertices)}
\begin{algleo}
\li \While NonEmpty $(Q)$ \Do
\begin{algleo}
\li $\langle z,D[z]\rangle \leftarrow$ Dequeue$(Q)$
\li \Foreach $e\in E$ s.t. $z\in e$
\begin{algleo}
\li \Foreach $v\in e$
\begin{algleo}
\li \If $D[v]>D[z]+w(e)$ \Then
\begin{algleo}
\li $D[v]\leftarrow D(z)+w(e)$; $P[v]\leftarrow z$; $E[v]\leftarrow e$
\li Enqueue or Update$(Q,\langle v,D[v]\rangle)$
\end{algleo}
\li \End; \End; \End; \End
\end{algleo}
\end{algleo}
\end{algleo}
\end{algleo}
\end{algleo}

\begin{theorem}\label{thm:de}
If before the weight decrease, $D[v]=d[v]$, $E[v]$ and $P[v]$ are correct for all $v\in V$, then after the weight decrease, $D[v]=d'[v]$ and $E[v]$ and $P[v]$ are correctly updated.
\end{theorem}
\begin{proof}
See Appendix~A.
\end{proof}

\subsection{Hyperedge Weight Increase}

 The coloring process in the graph case relies on the tree structure of the shortest paths,
which is no longer present in the shortest hyperpaths. Our solution is to use a relationship tree
for the coloring process, and we prove the correctness of this approach regardless of
the choice of the relationship tree.

Consider that the weight of a hyperedge $\check{e}$ increases to $w_{new}$.
First, we redefine the color of a vertex $v$ based on the chosen relationship tree.
\begin{itemize}
\item[\textbf{(1)}]  $v$ is colored {\em white} if $d'[v]=d[v]$ while keeping the current $P[v]$ and $E[v]$.
\item[\textbf{(2)}]  $v$ is colored {\em pink} if $d'[v]=d[v]$, but only possible through a new $P[v]$ or $E[v]$ or both.
\item[\textbf{(3)}]  $v$ is colored {\em red} if $d'[v]<d[v]$.
\end{itemize}

With the above modified definitions of colors, the same coloring process as in the graph case can be carried out
using a relationship tree. The algorithm is given below.

\setcounter{algleo}{0}
\begin{algleo}
\linonumber {\bf HE-DSP: Weight Increase}$(\check{e},w_{new})$.
\linonumber {\bf Step0 (Update the hypergraph)}
\begin{algleo}
\li $w(\check{e})\leftarrow w_{new}$
\end{algleo}
\linonumber {\bf Step1 (Determine the affected vertices in $e$)}
\begin{algleo}
\li \Foreach $v\in \check{e}$ s.t. $E[v]=\check{e}$ \Do
\begin{algleo}
\li Enqueue$(M, \langle v,D[v]\rangle)$
\end{algleo}
\end{algleo}
\linonumber {\bf Step2 (Coloring process)}
\begin{algleo}
\li \While NonEmpty($M$)
\begin{algleo}
\li $\langle z, D[z]\rangle\leftarrow$ Dequeue($M$)
\li \If $\exists$ $nonred$ $q\in V$ s.t. $\exists e\in E$ with $q,z\in e$ and $D[q]+w(e)=D[z]$
\begin{algleo}
\li \Then z is pink; $P[z]=q$; $E[z]=e$;
\li \Else z is red; Enqueue($M$, all $z$'s children)
\end{algleo}
\li \End; \End
\end{algleo}
\end{algleo}
\linonumber {\bf Step3.a (Initialize the distance vector for red vertices)}
\begin{algleo}
\li \Foreach $red$ vertex $z$ \Do
\begin{algleo}
\li \If $z$ has no $nonred$ neighbor
\begin{algleo}
\li \Then $D[z]\leftarrow +\infty$; $P[z]\leftarrow$ Null
\li \Else
\begin{algleo}
\li let $u$ be the $best$ $nonred$ $neighbor$ of z
\li $E[z]\leftarrow \textrm{argmin}_{e\in E, e\ni u,z}\{w(e)\}$;
\li $D[z]\leftarrow D[u]+w(E[z])$; $P[z]\leftarrow u$;
\li Enqueue$(Q,\langle z,D[z] \rangle)$
\end{algleo}
\li \End; \End; \End
\end{algleo}
\end{algleo}
\end{algleo}
\linonumber {\bf Step3.b: Step2 of HE-DSP: Weight Decrease}
\end{algleo}

The theorem below states the correctness of the algorithm.
\begin{theorem} \label{thm:in}
If before the weight increase, $D[v]=d[v]$, $E[v]$ and $P[v]$ are correct for all $v\in V$, then after the weight increase, $D[v]=d'[v]$ and also $E[v]$ and $P[v]$ are correctly updated.
\end{theorem}
\begin{proof}
See Appendix~B.
\end{proof}

\section{Dimension Reduction based Dynamic Shortest Path Algorithm}
\label{sec:D-1SESP}

In this section, we propose DR-DSP. When the dynamic problem degenerates to the
static problem, DR-DSP leads to an alternative algorithm for solving the static shortest hyperpath problem.

\subsection{The Static Case: DR-SP}

We first consider the static version of the algorithm (referred to as DR-SP), which captures the
basic idea of dimension reduction.

The proposed DR-SP algorithm is based on the following theorem in which we show that for a
general hypergraph $H$, the weight $\omega(L^*)$ of the shortest path $L^*$ of $H$ is equal to the
shortest path $L_G^*$ of a weighted graph $G$ derived from $H$. Specifically, corresponding to every
hyperedge $e$ in $H$, $G$ contains a clique defined on the vertices of $e$.

\begin{theorem} \label{thm:SP}
Let $H=(V,E,w)$ be a hypergraph, and $G=(V,\tilde{E})$ the
underlying graph of $H$ where an edge $\te\in \tE$ if and only if $\exists e\in E$ such that $\te \subset e$.
For each edge $\te$ in $G$, its weight $w_G (\te)$ is defined as
\begin{eqnarray} \label{eqn:w_G_e}
w_G (\te) = \underset{\{e\in E:~e\supseteq \te\}}{\min}~w(e).
\end{eqnarray}
Let $L^*$ and $L^*_G$ be the shortest paths from $u \in V$ to
$v \in V$ in $H$ and $G$, respectively. Then we have that
\[
w(L^*)=w_G(L^*_G).
\]
\end{theorem}

\begin{proof}
First,  for each shortest path $L^*_G$ in $G$, we can obtain a corresponding hyperpath $L$ in $H$ with the same
weight based on (\ref{eqn:w_G_e}), therefore we have that
\[
w_G (L^*_G)=w(L)\geq w(L^*).
\]
Then it suffices to show that there exists a path $L_G$ in $G$ such
that $w_G (L_G)\leq w(L^*)$, which implies that $w_G (L^*_G)\leq w_G
(L_G)\leq w(L^*)$.

Assume that $L^* = \{e_0,e_1,\ldots,e_{k-1}\}$ is a shortest hyperedge
path from $v_0$ to $v_k$ in $H$ where $v_0 \in e_0$ and $v_k \in
e_{k-1}$. Let $v_i \in e_{i-1} \cap e_i$ $(i=1,2,...,k-1)$ be one of
the vertices in the intersection of hyperedges $e_{i-1}$ and
$e_i$. Construct a path $L_G = \{v_0,v_1,...,v_k\}$ in the graph $G$.
For each edge $\te_i = \{v_i,v_{i+1}\}$ $(i = 0,1,...,k-1)$, since
$\te_i \subseteq e_i$, it follows from (\ref{eqn:w_G_e}) that
\[
w_G (\te_i)\leq w(e_i).
\]
Thus,
\[
w_G (L_G) = \sum_{i=0}^{k-1} w_G (\te_i) \leq \sum_{i=0}^{k-1} w(e_i)
= w(L^*),
\]
\ie $w_G (L_G)\leq w(L^*)$.
\end{proof}
It follows from Theorem \ref{thm:SP} that the shortest path in a general hypergraph can be obtained by
applying Dijkstra's algorithm to the underlying graph $G$ with weights modified as stated in the theorem.

\subsection{The Dynamic Case: DR-DSP}

In the dynamic case, a sequence $C=\{\delta_1,\delta_2,\ldots,\delta_l\}$ of
hyperedge changes in the hypergraph $H$ results in a sequence 
of edge changes in the underlying graph $G$. For each hyperedge change $\delta_i$, DR-DSP first updates the underlying
graph $G$ to locate all the changed edges caused by $\delta_i$. In the next step, DR-DSP updates the shortest path tree
in the underlying graph $G$.

Consider first the graph update. A change to a hyperedge $e$ only affects those edges in $G$ that are
subsets of $e$, \ie a hyperedge change is localized in the underlying graph $G$. Furthermore, since the weight of an edge in $G$
is the minimum weight of all hyperedges containing it, not all edges in $G$ that are subsets of $e$ will change weight. Based on
these observations, we propose a special data structure and procedure for updating the underlying graph $G$ without regenerating
the graph from scratch using Step~1 of DR-SP.

At the initialization stage of the algorithm, a priority queue $M_{uv}$ for each pair of vertices $(u,v)$ in the hypergraph
is established to store the weights of all hyperedges that
contain both $u$ and $v$. When a change occurs to hyperedge $e$, all the priority queues $M_{uv}$ associated with the pair of vertices
$(u,v)$ that are contained in $e$ are updated with the new weight of $e$. Thus, the top of these priority queues always maintain the
weight for edge $(u,v)$ in the underlying graph $G$ for each $(u,v)$.
Below is a pseudo code implementation of the proposed procedure.

\setcounter{algleo}{0}
\begin{algleo}
\linonumber {\bf Graph Update$(\check{e},w_{new})$}.
\begin{algleo}
\li \Foreach $u,v\in \check{e}$
\begin{algleo}
\li Update($M_{uv},<\check{e},w_{new}>$);
\li $w_{uv}\leftarrow$Peek($M_{uv}$);
\end{algleo}
\li \End;
\end{algleo}
\end{algleo}

After the underlying graph $G$ is updated, we are now facing a dynamic shortest path problem in a graph. However, since a single hyperedge
change can result in multiple edge changes in $G$, we need to handle a batch problem. While existing batch algorithms and
iterative single-change algorithms for graphs can be directly applied here, we show that the batch problem we have at hand
has two unique properties that can be exploited to improve the efficiency of the algorithm.

\begin{property} \label{prt:oneside}
The edge changes in $G$ caused by a hyperedge change are either all weight decreases or all weight increases.
\end{property}

\begin{property} \label{prt:clique}
All changed edges in $G$ caused by a hyperedge change belong to a clique in $G$.
\end{property}

\subsection{Hyperedge Weight Decrease}
If the weight of hyperedge $\check{e}$ decreases to $w_{new}$ , by Theorem~\ref{thm:SP} and Property~\ref{prt:oneside},
there are (possibly) several edge-weight decreases in the underlying graph $G$. Therefore similar to HE-DSP,
there is at least one unaffected node $x=\textrm{argmin}_{v\in \check{e}}\{D[v]\}$. By Property \ref{prt:clique}, these
affected edges are contained in a clique derived from the changed hyperedge; therefore it is sufficient to determine the distance
of every node $v$ (other than $x$) in the original changed hyperedge $e$ by checking $D[x]+w_{new}<D[v]$. And we can initialize
the priority queue with those nodes whose weight decreases. After that, the procedure is similar to that in the graph case.

\setcounter{algleo}{0}
\begin{algleo}
\linonumber {\bf DR-DSP: Weight Decrease$(\check{e},w_{new})$}.
\linonumber {\bf Step0 (Update the hypergraph and $G$)}
\begin{algleo}
\li $w(\check{e})\leftarrow w_{new}$
\li {\bf Graph Update}($\check{e},w_{new}$)
\end{algleo}
\linonumber {\bf Step1 of HE-DSP: Weight Decrease}
\linonumber {\bf Step2 of Graph: Weight Decrease}
\end{algleo}

\subsection{Hyperedge Weight Increase}
If the weight of hyperedge $\check{e}$ increases to $w_{new}$, by Theorem~\ref{thm:SP} and Property~\ref{prt:oneside}, there are (possibly) several
edge-weight increases in the underlying graph $G$. Similar to the single-change case in graph, there is at least
one unaffected node $x=\textrm{argmin}_{v\in \check{e}}\{D[v]\}$. Then another node $v\in \check{e}$ is affected only if $E[v]=\check{e}$, \ie
$\check{e}$ is on its shortest hyperpath. We use all such nodes to initialize the priority queue $M$. The rest is similar to
the procedure of Graph: Weight Increase.

\setcounter{algleo}{0}
\begin{algleo}
\linonumber {\bf DR-DSP: Weight Increase$(\check{e},w_{new})$}.
\linonumber {\bf Step0 (Update the hypergraph and $G$)}
\begin{algleo}
\li $w(\check{e})\leftarrow w_{new}$
\li {\bf Graph Update}($\check{e},w_{new}$)
\end{algleo}
\linonumber {\bf Step1 of HE-DSP: Weight Increase}
\linonumber {\bf Step2 of Graph: Weight Increase}
\linonumber {\bf Step3.a of Graph: Weight Increase}
\linonumber {\bf Step3.b of Graph: Weight Increase}
\end{algleo}

\section{Time Complexity Analysis}
\label{sec:PC}

We analyze the time complexity of the two proposed dynamic algorithms. We show that for different scenarios,
each algorithm has its own advantage. We also consider the static case and show that the static version of DR-DSP has the same
complexity as Gallo's algorithm for a general hypergraph and lower complexity for a simplicial complex.

\subsection{The Static Shortest Hyperpath Problem}

Given a hypergraph $H=(V,E,w)$, let $n=|V|$ denote the number of vertices in $H$, and
$\Phi=\sum_{e\in E} |e|^2$ where $|e|$ is the cardinality of $e$. For a simplicial complex, let $m$ be
the number of facets, and $d$ the maximum degree of the facets.
\begin{theorem} \label{thm:STC}
The time complexities of Gallor's algorithm and DR-SP for general hypergraphs and simplicial complexes are as follows.
\begin{table}[h!]
\centering
\begin{tabular}{||c|c|c||}
\hline\hline
Algorithm & General Hypergraph & Simplicial Complex \\
\hline
Gallo & $O(n\log n+\Phi)$ & $O(n\log n+d^22^dm)$\\
\hline
DR-SP & $O(n\log n+\Phi)$ & $O(n\log n+d2^dm)$\\
\hline\hline
\end{tabular}
\end{table}\
\end{theorem}
\begin{proof}
The time complexity of DR-SP mainly comes from Steps~1 and 2. Step~2 is essentially applying Dijkstra's algorithm to a
graph with $n$ vertices and $\tilde{m}$ edges where $\tilde{m}$ is the number of edges in the underlying graph $G$.
The running time is thus $O(n\log n+\tilde{m})$. An implementation of Step~1
is to obtain the edge weight $w_G(\te)$ based on (\ref{eqn:w_G_e}). Therefore the time complexity for Step~1
is $O(\sum_{e\in E}|e|^2)$, \ie $O(\Phi)$. With $\tilde{m}$ upper bounded by $\Phi$ (since for each $e\in E$,
there are at most $|e|(|e|-1)/2$ edges in $G$), we arrive at the total time complexity of DRSP.

For Gallo's Algorithm, similar to Dijkstra's algorithm, the time complexity is mainly in updating the neighbors of the non-fixed vertex $z$
with the minimal distance $D[z]$. For each $z$, the algorithm scans all the hyperedges containing $z$.
For each pair of vertices $(u,v)\in e$, $e$ is scanned twice. Therefore the total number of such operations is $\Phi=\sum_{e\in E}|e|^2$.
Also, extracting $z$ from the priority queue implemented by a fibonacci heap takes $O(\log n)$ time. The total time complexity
of Gallo's algorithm thus follows.

For a simplicial complex, $\Phi=O(d^22^dm)$, the complexity of Gallo's algorithm thus follows.
For DR-SP, exploiting the property that the edge set is closed under the subset operation in a simplicial complex,
we can use a top-down scheme in Step~1 of DR-SP to calculate the weight $w_G (s)$ inductively
with respect to the dimension of a facet as follows:
\[
w_G (s) = \min\{w(s),\{w_G (s')|~s' \supset s \textrm{ and dim}[s'] = i+1\}\},
\]
where $w_G (s') = w(s')$ for the facet $s'$.  The time complexity for Step~1 can then be improved.
Because each $i-$dimensional face is associated with $d-i$ comparisons. Thus, the running time of Step~1 for each $d$-dimensional facet is given by
\[
\sum_{i=1}^{d-1} (^{d+1}_{i+1})(d-i)=O(d2^d).
\]
Therefore the time complexity for Step~1 is $O(d2^dm)$. The total time complexity thus follows.
\end{proof}

\subsection{The Dynamic Shortest Hyperpath Problem}

Given a hypergraph $H=(V,E,w)$ and a change to hyperedge $e$, let
$|\delta|$ denote the number of affected vertices, $\|\delta\|$ the number of affected hyperedges
plus $|\delta|$,  $|\delta_\Phi|=\sum_{e\in E, e \textrm{ is affected}}|e|^2$, and
$\|\tilde{\delta}\|$ the number of affected edges in the underlying graph plus $|\delta|$.
\begin{theorem} \label{thm:DY}
The time complexities of HE-DSP and DR-DSP for the fully dynamic shortest path problem in a general hyperpath are as follows.
\begin{table}[h!]
\centering
\begin{tabular}{||c|c||}
\hline\hline
Algorithm & Time Complexity \\
\hline
HE-DSP & $O(|\delta|\log |\delta|+|\delta_\Phi|)$\\
\hline
DR-DSP & $O(|\delta|\log|\delta|+\|\tilde{\delta}\|+|e|^2\log m)$\\
\hline\hline
\end{tabular}
\end{table}\
\end{theorem}
\begin{proof}
For HE-DSP: Weight Decrease, the dominating part is Step~2. In Step~2, there are total $|\delta|$ iterations.
In each iteration, the algorithm first dequeues one node $z$ from $M$ which takes $O(\log |\delta|)$ time.
Then the algorithm updates all of $z$'s neighbors by scanning all the hyperedges containing $z$. Each affected
hyperedge $e$ can be scanned at most $|e|(|e|-1)=O(|e|^2)$ times. Therefore the time spent on updates for all
iterations is $O(|\delta_\Phi|)$. The total time complexity of HE-DSP: Weight Decrease is $O(|\delta|\log |\delta|+|\delta_\Phi|)$.
For HE-DSP: Weight Increase, similar to the above analysis, the time complexity for the dominating part (Step~2,
Step~3.a and Step~3.b) is $O(|\delta|\log |\delta|+|\delta_\Phi|)$. The total time complexity of HE-DSP: Weight Increase
is $O(|\delta|\log |\delta|+|\delta_\Phi|)$. The result thus follows.

For DR-DSP: Weight Decrease, the total time spent on Graph Update procedure is $O(|e|^2\log m)$.
In Step~2, there are $|\delta|$ iterations; in each iteration $O(\log |\delta|)$
time is spent to dequeue $z$ from $M$. Time spent on updating neighbors over all iterations is $O(\|\tilde{\delta}\|)$.
Therefore the total time complexity is $O(|\delta|\log|\delta|+\|\tilde{\delta}\|+|e|^2\log m)$.
For DR-DSP: Weight Increase,  Step~2, Step~3.a and Step~3.b take $O(|\delta|\log|\delta|+\|\tilde{\delta}\|)$ (similar
to the analysis for graphs). The total time
complexity thus follows.
\end{proof}

From Theorem \ref{thm:DY} we see that if $\delta$ is small and $|e|$ is large, HE-DSP
performs better, since in DR-DSP, the update of the underlying graph has to be done regardless
whether there are affected vertices. Thus in a sequence of hyperedge changes,
if only a small fraction of them actually have affected nodes, then HE-DSP will outperform DR-DSP.
On the other hind, if $\delta$ is large, because usually $|\delta_\Phi|\gg\|\tilde{\delta}\|$, then DR-DSP will outperform HE-DSP.
Consider the extreme example where every valid hyperedge exists, all nodes are affected and the changed hyperedge contains
$n$ vertices. Then $|\delta|=n$, $|\delta_\Phi|=O(n^22^n)$, $\|\tilde{\delta}\|=O(n^2)$, $|e|=n$, $m=O(2^n)$.
The time complexity of HE-DSP is $O(n\log n+n^22^n)=O(n^22^n)$ while the time complexity of DR-DSP is $O(n^3)$. We see
that the time complexity of DR-DSP can be much lower than that of HE-DSP.

\section{Simulation Results}
\label{sec:simu}

We present simulation results on the
running time of the proposed dynamic shortest hyperpath algorithms.
We test the proposed algorithms on hypergraphs generated from a random geometric model as well as those
generated by the Enron email data set. All simulation code is compiled
and run on the same laptop equipped with a 3.0GHz i7-920XM
Mobile Processor.

\subsection{Random Geometric Hypergraph}
We first consider a random geometric hypergraph model in which $n$ nodes are uniformly distributed in an $a\times a$ square.
All nodes within a circle with radius $r$ form a hyperedge (circles are centered on a $h\times h$ grid).
The weight of each hyperedge is given by the average distance
between all pairs of vertices of this hyperedge.

A sequence of changes are then generated and the proposed dynamic algorithms are employed to maintain all the shortest hyperpaths
from the source $s$ located at a corner of the $a\times a$ square.
Each change can be a hyperedge insertion (with probability $p_I$), a hyperedge deletion (with probability $p_D$), or a weight
change (with probability $1-p_I-p_D$) with new weight chosen uniformly in $[w_{min},w_{max}]$. In the case of a hyperedge deletion
or a weight change, the hyperedge to be deleted or to be assigned with a new weight is chosen according to the two models detailed below.
Hyperedge insertions are only \emph{realized} when there are hyperedges that have been deleted, and a randomly chosen one is inserted back.
This ensures that all hyperedges satisfy the geometric property determined by $r$ at all time. It also models the practical scenario
where a broken link is repaired.

In selecting a hyperedge for deletion or weight change, we consider two different models: the random change model and the targeted change model.
In the former, the hyperedge is randomly and uniformly chosen among all hyperedges. In the latter, it is randomly and uniformly chosen
from the current shortest hyperpaths. This models the scenarios where hyperedges in the shortest hyperpaths
are more prone to changes due to attacks, frequent use, or higher priority in maintenance and upgrade.

In Fig.~\ref{fig:D1}, we show the simulation results on the running time of the two proposed algorithms under a sequence
of $10^4$ changes. We see that HE-DSP has lower complexity in networks with random
topological and weight changes (Fig.~\ref{fig:D1}-Left), whereas DR-DSP should be preferred in networks
with targeted changes (Fig.~\ref{fig:D1}-Right).
This partition of the application space can be explained from the structures of these two algorithms. Under the random
change model, a large fraction of changes do not result in changes in the current shortest hyperpaths. Such changes lead to little
computation in maintaining the shortest hyperpaths for both algorithms, but requires about the same amount of computation in the
Graph-Update step of DR-DSP for maintaining the underlying graph. On the other hand, under the targeted change model, all hyperedge
deletions and weight changes affect the shortest hyperpaths. Updating the shortest hyperpaths can be done more efficiently
in DR-DSP since it works on the underlying graph with a much smaller number of edges.

\begin{figure}[h!]
\centering
\scalefig{0.23} \epsfbox{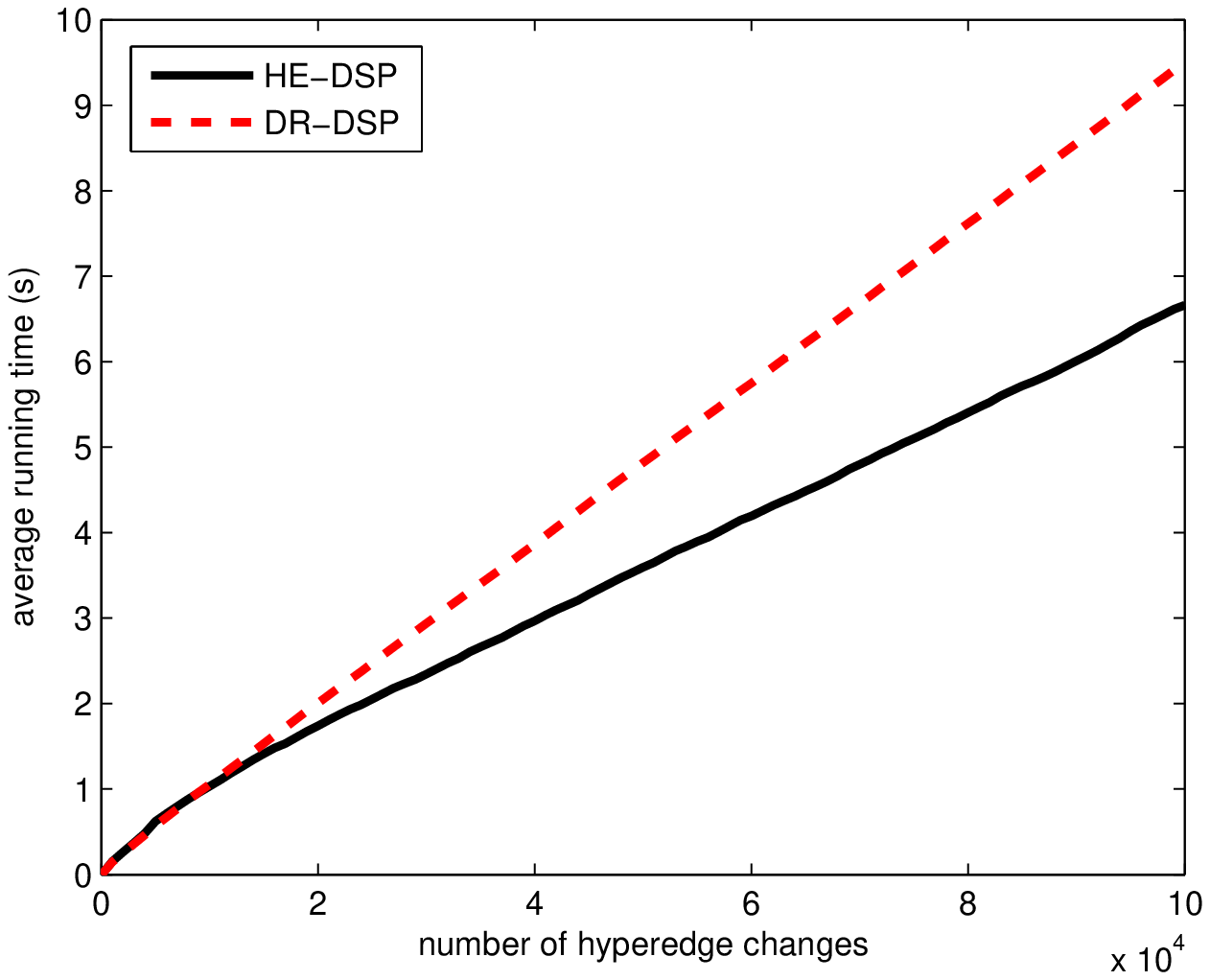}
\scalefig{0.23} \epsfbox{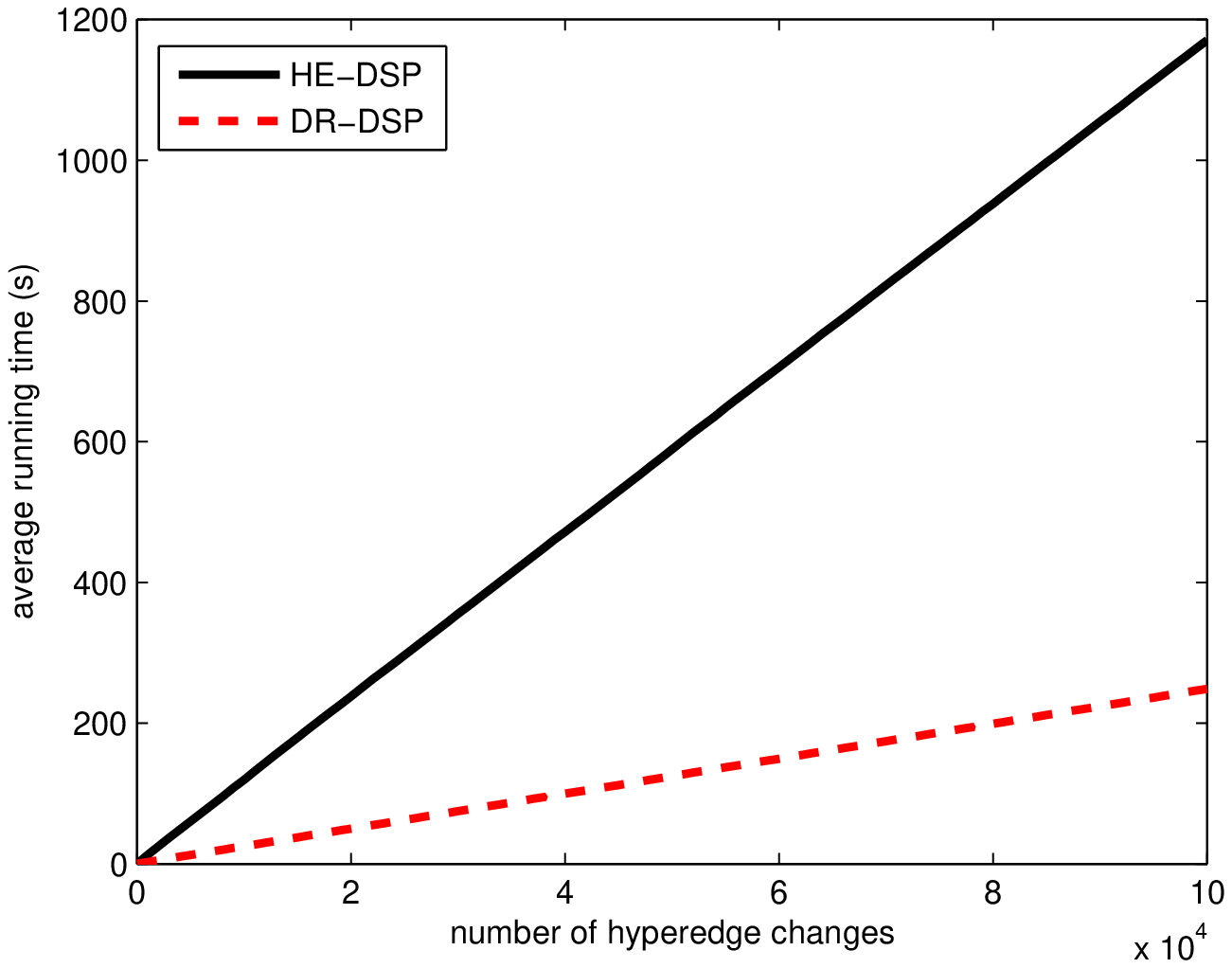}
\caption{The average running time. Left: the random change model; Right: the targeted change model ($n=1000$, $a=1000$, $r=\sqrt{1000}$, $h=1$,
$p_I=\frac{1}{4}$, $p_D=\frac{1}{4}$, $w_{\min}=10$, $w_{\max}=20$,
the average is taken over $50$ random hypergraphs).}
\label{fig:D1}
\end{figure}


\subsection{Enron Email Data Set}

In this example, we consider the application of the shortest hyperpath algorithms in finding the most important actor in
a social network. We consider the Enron email data set and use the same hypergraph generation model
as in~\cite{Y.Park2008}. Specifically, each person is a vertex of the hypergraph, and the sender and recipients of every email
form a hyperedge. Our objective is to identify the most important person measured by the closeness centrality index (\ie the total
weight of the shortest hyperpaths from this person to all the
other persons). The first step is to assign weight to each hyperedge that reflects ``distance''. While there is no universally
accepted way of measuring distance in a social network observed through email exchanges, certain general rules apply. First,
a direct email exchange between two persons indicates a stronger tie than an email sent to a large group. Thus, the weight of
an hyperedge should be an increasing function of the cardinality of this hyperedge. Second, more frequent email exchange
among a given group of people shows stronger ties. Thus, the weight of an hyperedge should be decreasing with the number of times
that this hyperedge appears in the email data set. Considering these two general rules, we adopt the following weight function:
\begin{eqnarray} \label{eqn:enron}
w(e) =(\sqrt{|e|})^{\alpha^{(l-1)}}
\end{eqnarray}
where $|e|$ is the cardinality of the hyperedge $e$,
$\alpha$ is the parameter measuring how fast the weight decreases with the number $l$ of times that this hyperedge appears in the data set.

We can then apply DR-SP on the resulting (static) hypergraph to find the shortest hyperpaths rooted at
each vertex and compute this vertex's closeness centrality index. With the weight function given in~\eqref{eqn:enron} using $\alpha=0.6$,
the identified most important actor is the CEO of Enron. The average distance (along the shortest hyperpath) from
the CEO to the other persons at various positions is shown in~Fig.~\ref{fig:E2}. We observe that in general, the higher the position, the shorter
the distance. These results demonstrate that the adopted hypergraph model and weight function capture the essence of the problem.
\begin{figure}[h!]
\centering
\scalefig{0.4} \epsfbox{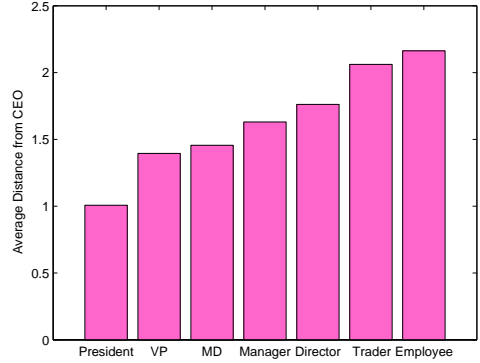}
\caption{The average distance from the CEO to others at different positions.}
\label{fig:E2}
\end{figure}

Next, we construct a dynamic hypergraph sequence based on the Enron data set. At the beginning, the hypergraph contains only
individual vertices. We then consider each email chronologically. Each email either adds a new hyperedge or decrease the weight of
an existing hyperedge (due to the increased number of appearances of this hyperedge). The two proposed algorithms are employed to
maintain the shortest hyperpaths rooted at the CEO after each change. The running time is given in Fig.~\ref{fig:E1}, which shows
the lower complexity of DR-DSP. The reason is that a large fraction of hyperedge changes result in changes in the shortest hyperpaths.

\begin{figure}[h!]
\centering
\scalefig{0.35} \epsfbox{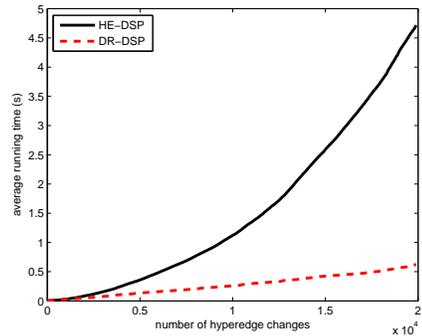}
\caption{The average running time for the Enron data set ($\alpha=0.6$, the average is taken over $50$ monte carlo runs).}
\label{fig:E1}
\end{figure}

\section{Conclusion}
We have presented, to our best knowledge, the first study of the fully dynamic shortest path problem in a general hypergraph.
We have developed two dynamic algorithms for finding and maintaining the shortest hyperpaths. These two algorithms
complement each other with each one preferred
in different types of hypergraphs and network dynamics, as illustrated in the time complexity analysis and simulation experiments.
We have discussed and studied via experiments over a real data set the potential applications of the dynamic shortest hyperpath problem
in  social and communication networks.

\section*{Appendix A: Proof of Theorem~\ref{thm:de}}\label{app:A}
The proof is based on the following three lemmas.
\begin{lemma} \label{lma:de1}
Let $x=\textrm{argmin}_{v\in\check{e}}\{d[v]\}$, then $d[x]=d'[x]$ and $d'[x]=\min_{v\in\check{e}}\{d'[v]\}$.
\end{lemma}
\begin{proof}
 Proof by contradiction. Assume that $d'[x]<d[x]$, then $x$ has to use $\check{e}$ on its new shortest hyperpath. Since we consider only simple hyperpaths and $x\in\check{e}$, we have $E[x]=\check{e}$. Therefore its parent $y=P[v]$ cannot use $\check{e}$ on its shortest hyperpath, which implies that the shortest distance to $y$ does not change: $d[y]=d'[y]$. Given that $y$ is the parent of $x$ on its new shortest hyperpath, we have $d[y]=d'[y]\leq d'[x]<d[x]$ which contradicts to the definition of $x$.

 For the second statement, assume there exists $z\in\check{e}$ such that $d'[z]<d'[x]$. Based on the definition of $x$ and the hypothetical assumption, $d[z]\geq d[x]=d'[x]>d'[z]$. It thus follows that $z$'s shortest hyperpath changes and $E[z]=\check{e}$ in the new shortest hyperpath.  Follow the same line of arguments by considering the parent of $z$, we arrive at the same contradiction in terms of the definition of $x$.
\end{proof}
\begin{lemma} \label{lma:de2}
For any vertex $v$, $v$ is enqueued into $Q$ if and only if $d'[v]<d[v]$.
\end{lemma}
\begin{proof}
Consider first that $v$ is enqueued into $Q$. From the algorithm, this can only happen if there exists a neighbor $z$ and a hyperedge $e\ni v,z$ such that $D[z]+w(e)<D[v]$. We thus have $d[v]\geq D[v]>D[z]+w(e)\geq d'[v]$ (note that at any time, $d[v]\ge D[v] \ge d'[v]$, which can be easily seen from the procedure of the algorithm).

We now prove the converse. Assume that $d'[v]<d[v]$. Let $p=\{e_1,e_2,\ldots,e_i,\check{e},e_{i+1},\ldots,e_l\}$ be $v$'s new shortest hyperpath. There exists $u_{i+1}\in \check{e}\cap e_{i+1}$ such that $d'[u_{i+1}]<d[u_{i+1}]$. In Step~1 of the algorithm, $u_{i+1}$ is enqueued. Similarly, there exists $u_{i+2}\in e_{i+1}\cap e_{i+2}$ with $d'[u_{i+2}]<d[u_{i+2}]$. Then $u_{i+2}$ will be enqueued in Step~2 of the algorithm when $u_{i+1}$ is dequeued if it has not been enqueued before that. Repeating this line of argument, we conclude that there exits $u_l\in e_{l-1} \cap e_l$ with $d'[u_l]<d[u_l]$ and $u_l$ is enqueued into $Q$. Then $v$ will be enqueued when $u_l$ is dequeued if it is not enqueued already.
\end{proof}
\begin{lemma} \label{lma:de3}
For each $v$ dequeued from $Q$, $D[v]=d'[v]$.
\end{lemma}
\begin{proof}
We first show that if $u$ is dequeued before $v$, then $D[u]\le D[v]$ at the instants when they are dequeued. We prove this by induction. The initial condition holds trivially. Then assume it is true for the first $l$ dequeued vertices  $z_1,\ldots,z_l$. Consider the $(l+1)$th dequeued vertex $z_{l+1}$. At the instant when $z_l$ is dequeued, if $D[z_{l+1}]$ is updated based on $D[z_l]$ in Step~2, then $D[z_l] < D[z_{l+1}]$ even after the update. If, on the other hand,  $D[z_{l+1}]$ is not updated at this instant, then
$D[z_l] \le D[z_{l+1}]$ given that the dequeued vertex has the smallest distance.

Next, we prove the lemma by induction. From Step~1 of the algorithm, all the affected vertices $v$ in $\check{e}$ will be dequeued first with $E[v]=\check{e}$, $P[v]=x$, and $D[v]=d'[x]+w(\check{e})$. Based on Lemma~\ref{lma:de1}, $D[v]\le d'[u] +  w(\check{e})$ for any $u\in\check{e}$. It thus follows that the hyperpath to $v$ through $x$ and $\check{e}$ is the shortest one with $D[v]=d'[v]$.

Assume for $z_1,\ldots,z_l$, $D[z_i]=d'[z_i]$ are satisfied for all $i=1,\ldots,l$. Consider the $(l+1)$th dequeued vertex $z_{l+1}\not\in \check{e}$. Let $u=P[z_{l+1}]$ be its parent in the new shortest hyperpath. Then based on the fact that distances of the dequeued vertices are monotonically increasing with the order of the dequeueing as shown at the beginning of the proof, $u$ cannot be any vertex dequeued after $z_{l+1}$. Since $z_{l+1}\not\in\check{e}$, it is also clear that $u$ cannot be an unaffected vertex(otherwise, $z_{l+1}$ will be unaffected, which contradicts Lemma~\ref{lma:de2}). We thus have $u\in\{z_1,\ldots,z_l\}$. Let $u=z_i$. Then when $z_i$ is dequeued , $D[z_{l+1}]$ will be updated to the shortest distance $d'[z_{l+1}]$ due to the induction hypothesis of $D[z_i]=d'[z_i]$. This completes the proof.
\end{proof}
Based on Lemma~\ref{lma:de2} and \ref{lma:de3}, the shortest distances of all affected vertices will be updated correctly. Based on Lemma~\ref{lma:de2}, all unaffected vertices will not be enqueued, and their distances remain the same. It is not difficulty to see from the algorithm that $P[v]$ and $E[v]$ are also correctly maintained for all $v$.
\section*{Appendix B: Proof of Theorem~\ref{thm:in}}\label{app:B}
We first show the correctness of the coloring process as given in the following lemma.
\begin{lemma}\label{lma:in1}
The coloring process correctly
colors all the affected vertices.
\end{lemma}
\begin{proof}
We first state the following simple facts without proof: given a relationship tree, after the hyperedge weight increase, (1) if $v$ is pink or white, then all its descendent in this relationship tree are white; (2) if a $v$ is red, then all its children in the relationship tree are either pink or red; (3) if a $v$ is affected, either $v\in \check{e}$ or $P[v]$ is red. These facts can be directly obtained from the definition of the color. It is also easy to see that vertices are dequeued from $M$ in a nondecreasing order of their current distance $D[\cdot]$. This is because each time a vertex $z$ is dequeued from $M$, the possible new vertices to be enqueued into $M$ are $z$'s children with distances no smaller than $D[z]$.

Then, the proof of the lemma has two parts: first we prove that all affected vertices are enqueued into $M$; then we prove by induction that only affected vertices are enqueued into $M$ and their colors are correctly identified.

We prove the first part by contradiction. Assume that there exists an affected vertex $v$ that is not enqueued into $M$. It is easy to see that $v\not\in\check{e}$ because all the affected vertices in $\check{e}$ are enqueued in Step~1. Based on the third fact stated above, $P[v]$ is red. Based on the hypothesis, $P[v]$ is not enqueued (otherwise, $v$ will be enqueued in Step~2). Continue this line of arguments, we eventually reach the root of the relationship tree and arrive at the contradiction that the source $s$ is red.

We prove the second part by induction. It is easy to see that all the vertices initially enqueued into $M$ are affected vertices. It remains to show that the first vertex $z_1$ dequeued from $M$ is colored (pink or red) correctly. To show that, we need to establish that the algorithm correctly determines whether there is an alternative shortest hyperpath to $z_1$ with the same distance, \ie $d[z_1]=d'[z_1]$. The key here is to show that checking the currently non-red neighbors (which may become red in the future) of $z_1$ will not lead to a false alternative path. This follows from the fact that $z_1$ has the smallest distance $D[\cdot]$ among all affected vertices (which belong to the set of vertices consisting of the affected vertices in  $\check{e}$ and their descendents).

Next, assume that vertices $z_1, z_2,\ldots,z_l$ dequeued from $M$ are all affected vertices and are correctly colored. Consider the next dequeued vertex $z_{l+1}$. It is an affected vertex because it is either enqueued in Step~1 with $E[v]=\check{e}$ or enqueued in Step~2 with a red parent. To show that $z_{l+1}$ will be colored correctly, we use a similar argument by showing that the currently non-red neighbors of $z_{l+1}$ will not give a false alternative path. The latter follows from the fact that all affected vertices will be enqueued and those dequeued after $z_{l+1}$ have distances no smaller than $D[z_{l+1}]$. This completes the induction.
\end{proof}

We now show that $D[v],~P[v]$ and $E[v]$ are correctly maintained for all $v$. For each red vertex $v$, its distance is set based on the current shortest distance from a non-red neighbor in Step~3.a. The rest of the algorithm is essentially Gallo's extension of Dijkastra's algorithm with the current initial distance. The correctness of the algorithm thus follows. It is not difficult to see that $P[\cdot]$ and $E[\cdot]$ are correctly updated for both red and pink vertices.

\bibliographystyle{ieeetr}

\end{document}